\documentclass[11pt,a4paper]{article}

\usepackage{fullpage}
\newcommand{\ignore}[1]{}

\ignore{
\usepackage[
normalsections, %Don~Rt modify section headers.
normalmargins, %Don~Rt modify page margins.
normallists, %Don~Rt modify the itemize, enumerate, or description environments.
normalfloats, %Don~Rt modify LATEX2"~Rs float parameters.
normalindent, %Don~Rt modify paragraph indentation.
normaltitle, %Don~Rt modify the formatting of the document title.
normalleading, %Don~Rt modify interline spacing.
normallooseness, %Don~Rt modify paragraph looseness.
]{savetrees}
\usepackage{xspace}   % deletes the annoying {} at the ends of defs.
}

\usepackage{amsthm}

\newtheorem {theorem}{Theorem}

\newtheorem {lemma}[theorem]{Lemma}
\newtheorem {corollary}[theorem]{Corollary}

\date{}

\usepackage{enumerate}
\usepackage{eucal}
\usepackage{amsmath,amstext,amssymb}

\newcommand{\defparproblemu}[4]{
  \vspace{1mm}
\noindent\fbox{
  \begin{minipage}{0.97\textwidth}
  \begin{tabular*}{0.97\textwidth}{@{\extracolsep{\fill}}lr} #1 & {\bf{Parameter:}} #3 \\ \end{tabular*}
  {\bf{Input:}} #2  \\
  {\bf{Question:}} #4
  \end{minipage}
  }
  \vspace{1mm}
}

\newcommand{\pcvcname}{{\sc{Planar Connected Vertex Cover}}}
\newcommand{\pcondom}{{\sc{Planar Connected Dominating Set}}}

\newcommand{\indleafname}{{\sc{Maximum Independent Leaf Spanning Tree}}}
\newcommand{\maxleafname}{{\sc{Max Leaf}}}
\newcommand{\pmaxleafname}{{\sc{Planar Max Leaf}}}
\newcommand{\nsisname}{{\sc{Maximum Nonseparating Independent Set}}}
\newcommand{\pnsisname}{{\sc{Planar Maximum Nonseparating Independent Set}}}
\newcommand{\ch}{{\rm ch}}
\newcommand{\C}{\mathcal{C}}
\newcommand{\case}[1]{\noindent{\bf CASE #1}}
\renewcommand{\P}{\mathbf{P}}
\newcommand{\NP}{\mathbf{NP}}
\newcommand{\myparagraph}[1]{\vskip 2mm \noindent{\bf #1\ }}

\begin{document}

\title{A $9k$ kernel for nonseparating independent set in planar graphs\thanks{Work supported by the National Science Centre (grant N N206 567140). An extended abstract was presented at WG 2012.}}

\author{
  \L{}ukasz Kowalik, Marcin Mucha\\
  Institute of Informatics, University of Warsaw, Poland\\
  \texttt{\{kowalik,mucha\}@mimuw.edu.pl}
}

\maketitle

\begin{abstract}
  We study kernelization (a kind of efficient preprocessing) for NP-hard problems on planar graphs.
  Our main result is a kernel of size at most $9k$ vertices for the \pnsisname\  problem. 
  A direct consequence of this result is that \pcvcname\ has no kernel with at most $(9/8-\epsilon)k$ vertices, for any $\epsilon>0$, assuming $\P\ne\NP$.
  We also show a very simple $5k$-vertices kernel for \pmaxleafname, which results in a lower bound of $(5/4-\epsilon)k$ vertices for the kernel of \pcondom\ (also under $\P\ne \NP$).
  
  As a by-product we show a few extremal graph theory results which might be of independent interest.
  We prove that graphs that contain no separator consisting of only degree two vertices contain
  (a) a spanning tree with at least $n/4$ leaves and (b) a nonseparating independent set of size at least $n/9$ (also, equivalently, a connected vertex cover of size at most $\frac{8}{9}n$). The result (a) is a generalization of a theorem of Kleitman and West~\cite{kw:many-leaves} who showed the same bound for graphs of minimum degree three. Finally we show that every $n$-vertex outerplanar graph contains an independent set $I$ and a collection of vertex-disjoint cycles $\C$ such that $9|I|\ge 4n-3|\C|$. 
\end{abstract}

\section{Introduction}

Many NP-complete problems, while most likely not solvable efficiently, admit kernelization algorithms, i.e.\ efficient algorithms which replace input instances with an equivalent, but often much smaller one. More precisely, a {\em kernelization algorithm} takes an instance $I$ of size $n$ and a parameter $k\in\mathbb{N}$, and after time polynomial in $n$ it outputs an instance $I'$ (called a {\em kernel}) with a parameter $k'$ such that $I$ is a yes-instance iff $I'$ is a yes instance, $k'\le k$, and $|I'|\le f(k)$ for some function $f$ depending only on $k$. The most desired case is when the function $f$ is polynomial, or even linear (then we say that the problem admits a polynomial or linear kernel). In such a case, when the parameter $k$ is relatively small, the input instance, possibly very large, is ``reduced'' to a small one. In this paper by the size of the instance $|I|$ we always mean the number of vertices.

In the area of kernelization of graph problems the class of planar graphs (and more generally $H$-minor-free graphs) is given special attention. This is not only because planar graphs are models of many real-life networks but also because many problems do not admit a (polynomial) kernel for general graphs, while restricted to planar graphs they have a polynomial (usually even linear) kernel. A classic example is the $335k$-vertex kernel for the {\sc Planar Dominating Set} due to Alber et al.~\cite{afn:planar-domset}. In search for optimal results, and motivated by practical applications, recently researchers try to optimize the constants in the linear function bounding the kernel size, e.g.\ the current best bound for the size of the kernel for the {\sc Planar Dominating Set} is $67k$~\cite{cfkx:duality-and-vertex}. Such improvements often require nontrivial auxiliary combinatorial results which might be of independent interest. Our paper fits into this framework. 

Let $G=(V,E)$ be a graph and let $S$ be a subset of $V$. By $G[S]$ we denote the subgraph of graph $G$ induced by a set of vertices $S$. 
We say that $S$ is {\em nonseparating} if $G[V-S]$ is connected, otherwise we say that $S$ is a {\em separator}.
We focus on kernelization of the following problem:

\defparproblemu{\nsisname\ (NSIS)}{a graph $G=(V,E)$ and an integer $k\in \mathbb{N}$}{$k$}{Is there a nonseparating independent set of size at least $k$?}

In what follows, $|V|$ is denoted by $n$.
This problem is closely related with {\sc Connected Vertex Cover} (CVC in short), where given a graph $G=(V,E)$ and an integer $k$ we ask whether there is a set $S\subseteq V$ of size at most $k$ such that $S$ is a vertex cover (i.e.\ every edge of $G$ has an endpoint in $S$) and $S$ induces a connected subgraph of $G$.
The CVC problem has been intensively studied, in particular there is a series of results on kernels for planar graphs~\cite{guon:planarkernels,mfcs} culminating in the recent $\tfrac{11}{3}k$ kernel~\cite{moja-cvc-11/3}.
It is easy to see that $C$ is a connected vertex cover iff $V-C$ is a nonseparating independent set. In other words, $(G,k)$ is a yes-instance of CVC iff $(G,n-k)$ is a yes-instance of NSIS. In such a case we say that NSIS is a {\em parametric dual} of CVC. An important property of a parametric dual, discovered by Chen et al~\cite{cfkx:duality-and-vertex}, is that if the dual problem admits a kernel of size at most $\alpha k$, then the original problem has no kernel of size at most $(\alpha/(\alpha-1)-\epsilon)k$, for any $\epsilon>0$, unless P$=$NP.

As we will see, the NSIS problem in planar graphs is strongly related to the \maxleafname\ problem: given a graph $G$ and an integer $k$, find a spanning tree with at least $k$ leaves. 

\myparagraph{Our Kernelization Results}
We study \pnsisname\ ({\sc Planar NSIS} in short), which is the NSIS problem restricted to planar graphs. We show a kernel of size at most $9k$ for {\sc Planar NSIS}. This implies that \pcvcname\ has no kernel of size at most $(9/8-\epsilon)k$, for any $\epsilon>0$, unless P$=$NP. This is the first non-trivial lower bound for the kernel size of the {\sc Planar CVC} problem. Our kernelization algorithm is very efficient: it can be implemented to run in $O(n)$ time. As a by-product of our considerations we also show a $5k$ kernel for both \maxleafname\ and \pmaxleafname, which in turn implies a lower bound of $(5/4-\epsilon)k$ for its parametric dual, i.e.\ \pcondom.

\myparagraph{Our Combinatorial Results}
Some of our auxiliary combinatorial results might be of independent interest. 
We prove that graphs that contain no separator consisting of only degree two vertices contain
(a) a spanning tree with at least $n/4$ leaves and (b) a nonseparating independent set of size at least $n/9$ (also, equivalently, a connected vertex cover of size at most $\frac{8}{9}n$). The result (a) is a generalization of a theorem of Kleitman and West~\cite{kw:many-leaves} who showed the same bound for connected graphs of minimum degree three. As an another variation of this theorem, we show that every connected graph in which every edge has an endpoint of degree at least 3 has a spanning tree with at least $n/5$ leaves.  
Finally we show that every $n$-vertex outerplanar graph contains an independent set $I$ and a collection of vertex-disjoint cycles $\C$ such that $9|I|\ge 4n-3|\C|$. 

\myparagraph{Previous Results}
There are a few interesting results on NSIS in cubic/subcubic graphs. Speckenmeyer~\cite{Speckenmeyer} and Huang and Liu~\cite{huang} focus on cubic graphs and show some interesting relations between the size of maximum nonseparating independent set and the values of maximum genus and the size of the minimum feedback vertex set. Ueno et al.~\cite{ueno} showed that the problem is polynomial-time solvable for subcubic graphs by a reduction to the matroid parity problem.

As the {\sc CVC} is NP-complete even in planar graphs~\cite{garey:planarcvc}, so is {\sc NSIS}.
To the best of our knowledge there is no prior work on the parameterized complexity of \nsisname. The reason for that is simple: a trivial reduction from {\sc Independent Set} (add a vertex connected to all the vertices of the original graph) shows that NSIS is W[1]-hard, i.e.\ existence of an algorithm of complexity $O(f(k)\cdot |V|^{O(1)})$ is very unlikely (and so is the existence of a polynomial kernel). However, by general results on kernelization for sparse graphs~\cite{fomin:bidim-kernels}, one can see that NSIS admits a $O(k)$ kernel for apex-minor-free graphs, so in particular for planar graphs. However, the general approach does not provide a good bound on the constant hidden in the asymptotic notation. Observe that this constant is crucial: since we deal with an NP-complete problem, in order to find an exact solution in the reduced instance, most likely we need exponential time (or at least superpolynomial, because for planar graphs $2^{O(\sqrt{k})}$-time algorithms are often possible), and the constant appears in the exponent. 

\maxleafname\ has been intensively studied. Although there is a $3.75k$ kernel even for general graphs due to Estivill-Castro et al.~\cite{fellows:maxleaf}, some of their reductions do not preserve planarity. Moreover, the algorithm and its analysis are extremely complicated, while our method is rather straightforward.

\myparagraph{Yet another equivalent formulation of NSIS}
Consider the NSIS problem again.
It is easy to see that if graph $G$ has two nontrivial (i.e.\ with at least two vertices) connected components, then the answer is NO. Furthermore, an instance $(G,k)$ consisting of a connected component $C$ and an independent set $I$ is equivalent to the instance $(G[C],k-|I|)$. Hence, w.l.o.g.\ we may assume that the input graph $G$ is connected. It is easy to see that the \nsisname\ problem for connected graphs is equivalent to the following problem, which we name \indleafname:

\defparproblemu{\indleafname}{a graph $G=(V,E)$ and an integer $k\in \mathbb{N}$}{$k$}{Is there a spanning tree $T$ such that the set of leaves of $T$ contains a subset of size $k$ that is independent in $G$?}

In what follows, we will use the above formulation, since it directly corresponds to our approach.

\myparagraph{Nonstandard terminology and notation}
%We use standard graph theory terminology, see~\cite{diestel}.
%In particular, 
By $N_G(v)$ we denote the set of neighbors of vertex $v$. 
By a $d$-vertex we mean a vertex of degree $d$.

\section{A simple $\boldmath{12k}$ kernel for Planar NSIS}
\label{sec:12k}

In this section we describe a relatively simple algorithm that finds a $12k$ kernel for {\sc Planar NSIS}. This is achieved by the following three steps. First, in Section~\ref{sec:rule} we show a reduction rule and a linear-time algorithm which, given an instance $(G,k)$, returns an equivalent instance $(G',k')$ such that $|V(G')|\le |V(G)|$, $k'\le k$, and moreover $G'$ has no separator consisting of only 2-vertices. Second, in Section~\ref{sec:many-leaves} we show that $G'$ has a spanning tree $T$ with at least $|V(G')|/4$ leaves (and it can be found in linear time). Denote the set of leaves of $T$ by $L$. Third, in Section~\ref{sec:outer} we show that the graph $G'[L]$ is outerplanar. It follows that $G'[L]$ has an independent set of size at least $|L|/3$ (which can be easily found in linear time) and, consequently, $T$ has at least $|V(G')|/12$ leaves that form an independent set. Hence, if $k'\le |V(G')|/12$ our algorithm returns the answer YES (and the relevant feasible solution if needed). Otherwise $|V(G')| < 12 k' \le 12k$ so $(G',k')$ is indeed the desired kernel.

\subsection{The separator rule}
\label{sec:rule}

Now we describe our main reduction rule, which we call {\em separator rule}.
It is easier for us to prove the correctness of the rule for the CVC problem and then convert it to a rule for the NSIS problem.

\begin{quote}
 {\bf Separator rule} Assume there is a separator $S$ consisting of only $2$-vertices.
 As long as $S$ contains two adjacent vertices, remove one of them from $S$ (note that $S$ is still a separator).
 Next, choose any $v\in S$ such that the two neighbors $a, b$ of $v$ belong to distinct connected components of $G[V-S]$. If $\deg(a)=\deg(b)=1$, remove $a$ from $G$. If $\deg(a)=1$ and $\deg(b)\ge 2$, remove $a$ from $G$ and decrease the parameter $k$ by 1. Proceed analogously when $\deg(b)=1$ and $\deg(a)\ge 2$. Finally, when $\deg(a),\deg(b)\ge 2$, contract the path $avb$ into a single vertex $v'$ and decrease $k$ by 2.
\end{quote}

We say that a reduction rule for a parameterized problem $P$ is {\em correct} when for every instance $(G,k)$ of $P$ it returns an instance $(G',k')$ such that:
\begin{enumerate}[a)]
\item $(G',k')$ is an instance of $P$,
\item $(G,k)$ is a yes-instance of $P$ iff $(G',k')$ is a yes-instance of $P$,
\item $k'\le k$.
\end{enumerate}

\begin{lemma}
\label{lem:reduce}
The separator rule is correct for {\sc Planar CVC}.
\end{lemma}

\begin{proof}
Since the separator rule modifies the graph by removing a vertex or contracting a path it is planarity preserving, so a) holds. The condition c) is easy to check so we focus on b), i.e.\ the equivalence of the instances.

The case when $\deg(a)=\deg(b)=1$ are trivial so we skip the argument.

Now assume $\deg(a)=1$, $\deg(b)\ge 2$ (the case $\deg(b)=1$, $\deg(a)\ge 2$ is symmetric). If $C$ is a minimum connected vertex cover of $G$, $|C|\le k$, then $v\in C$ and $a\not\in C$. Since $G[C]$ is connected, $\deg(v)=2$ and $\deg(b)\ge 2$, also $b\in C$. It follows that $C\setminus\{v\}$ is a connected vertex cover of $G'$ of size at most $k'=k-1$. In the other direction, if $C'$ is a connected vertex cover of $G'$, $|C'|\le k' = k-1$, then $b\in C'$ and clearly $C'\cup\{v\}$ is a connected vertex cover of $G$.
 
Finally, assume $\deg(a),\deg(b)\ge 2$. 
Let $A$ and $B$ be the connected components of $G[V-S]$ that contain $a$ and $b$, respectively.
Let $a_0$ (resp. $b_0$) be any neighbor of $a$ (resp. $b$) distinct from $v$ ($a_0$ and $b_0$ exist since $\deg(a),\deg(b)\ge 2$).
Note that $a_0 \in A \cup S$, $b_0\in B \cup S$ and $a_0,b_0 \in G'$. 

Let us first assume that $G'$ has a connected vertex cover $C'$, $|C'| \le k'$. 
We show that $G'$ has a connected vertex cover $D'$, $|D'| \le k'$ such that $v'\in D'$. Then it is easy to check that $C=(D'\setminus \{v'\}) \cup \{a,v,b\}$ is the required connected vertex cover of $G$, and $|C|=|D'|+2 \le k$.

If $v' \in C'$ we just put $D'=C'$ so suppose that $v' \not\in C'$. Then $a_0, b_0 \in C'$. Since $G'[C']$ is connected, there is a path $P$ from $a_0$ to $b_0$ in $G'[C']$, possibly of length 0. Since $a_0 \in A \cup S$ and $b_0\in B \cup S$ we infer that $P$ contains a vertex $w\in S\cap C'$.
It follows that $D' = C' \setminus \{w\} \cup \{v'\}$ is a connected vertex cover of $G'$ of size at most $k'$.

Let us now assume that $(G,k)$ has a connected vertex cover $C$, $|C| \le k$. If $\{a,v,b\} \subseteq C$, then clearly $C' = (C\setminus\{a,v,b\}) \cup \{v'\}$ is a connected vertex cover of $G'$ with $|C'| \le k-2 = k'$. On the other hand, if $\{a,v,b\} \not\subseteq C$, then $G[C \cup \{a,v,b\}]$ contains a cycle (because $C$ is connected). Since $S$ is a separator, this cycle has to contain some $w_1 \in S$ other than $v$. In this case we claim that $C' = (C \setminus \{w_1\}) \cup \{a,v,b\}$ is a connected vertex cover of $G'$ and $|C'| \le k'$. It is clear that $C'$ is a connected vertex cover. The size bound follows from the fact, that $C$ has to contain two out of the three vertices $\{a,v,b\}$. 
\ignore{
Let us now assume that $(G,k)$ has a connected vertex cover $C$, $|C| \le k$. 
We show that $G$ has a connected vertex cover $D$, $|D| \le k$ such that $\{a,v,b\}\subseteq D$. Then clearly $D\setminus\{a,v,b\} \cup \{v'\}$ is a connected vertex cover of $G'$ with $|C'| \le k-2 = k'$, as required. 

Note that $|\{a,v,b\}\cap C| \ge 2$. If $|\{a,v,b\}\cap C| =3$ we are done, so $|\{a,v,b\}\cap C| =2$.

Suppose $\{a,v,b\}\cap C = \{v,b\}$ (the case $\{a,v,b\}\cap C=\{v,a\}$ is symmetric).
Then $a_0 \in C$. Since $G[C]$ is connected, there is a path $P$ from $a_0$ to $b$ in $G[C]$. Since $a_0 \in A \cup S$ and $b\in B$ we infer that $P$ contains a vertex $w\in S\cap C$, $w\ne v$. Hence we can put $D=C\setminus\{w\}\cup\{a\}$.

Suppose $\{a,v,b\}\cap C = \{a,b\}$.
Since $G[C]$ is connected, there is a path $P$ from $a$ to $b$ in $G[C]$. Since $a \in A$ and $b\in B$ we infer that $P$ contains a vertex $w\in S\cap C$, $w\ne v$. Hence we can put $D=C\setminus\{w\}\cup\{v\}$.
} 
\end{proof}

Now, we convert the separator rule to the {\bf dual separator rule} as follows. Let $(G,\ell)$ be an instance of {\sc (Planar) NSIS}. Put $k=|V(G)|-\ell$, apply the separator rule to $(G,k)$ and get $(G',k')$. Put $\ell'=|V(G')|-k'$ and return $(G',\ell')$.

\begin{corollary}
\label{cor:reduce}
The dual separator rule is correct for {\sc Planar NSIS}.
\end{corollary}

\begin{proof}
 The condition c) is easy to check while Lemma~\ref{lem:reduce} implies a) and b). 
\end{proof}

It is clear that the dual separator rule can be implemented in linear time. 
However, we would like to stress a stronger claim: there is a linear-time algorithm that given a graph $G$ applies the separator rule as long as it is applicable. This algorithm can be sketched as follows. 
First we remove all 1-vertices that are adjacent to 2-vertices (note that this is a special case of the dual separator rule).
Second, we find all maximal paths that contain 2-vertices only. For every such path, if it contains at least 3 vertices (and hence two of them form a separator), we replace it by a path of two vertices (this is again a special case of the dual separator rule). It is easy to implement these two steps in linear time.
Now, every 2-vertex has at most one neighboring 2-vertex. 
We remove all 2-vertices that do not have neighbors of degree 2 and for each pair of adjacent 2-vertices we remove exactly one of them. Next, we pick a connected component $A$ of the resulting graph and we mark all its vertices. Then we consider all degree 2 neighbors of this component (that has been removed). If such a neighbor $v$ has an unmarked neighbor then it connects $A$ with another component $B$. We apply the separator rule to the vertex $v$ (in constant time) and we mark $v$ as {\em processed}. Then we mark all the vertices of $B$. As a result the components $A$ and $B$ are joined into a new component $A$. In any case, the vertex $v$ is not considered any more. We continue the procedure until the graph gets connected. All the removed 2-vertices that are not marked as processed are put back in the graph. Since every vertex of $G$ is marked at most once, the whole algorithm works in linear time.

\subsection{Finding a spanning tree with many leaves}
\label{sec:many-leaves}

Kleitman and West~\cite{kw:many-leaves} showed how to find a spanning tree with at least $n/4$ leaves in a graph of minimum degree 3. In this section we generalize their result by proving the following theorem.

\begin{theorem}
\label{thm:n/4}
Let $G$ be a connected $n$-vertex graph that does not contain a separator consisting of only 2-vertices.
Then $G$ has a spanning tree with at least $n/4$ leaves. Moreover, such a tree can be found in linear time.
\end{theorem}

We will slightly modify the approach of Kleitman and West so that vertices of smaller degree are allowed.

First note that it suffices to show a simplified case where $G$ has no edge $uv$ such that both $u$ and $v$ are 2-vertices (and we still assume the nonexistence of a separator consisting of only 2-vertices).
Indeed, if the theorem holds for the simplified case, we just remove from $G$ all the edges $uv$ such that both $u$ and $v$ are 2-vertices, call the new graph $\tilde{G}$. Note that $\tilde{G}$ is connected and $\tilde{G}$ does not contain a separator consisting of only 2-vertices (otherwise the old $G$ contains such a separator). Hence we get a spanning tree $T$ of $\tilde{G}$ with at least $|V(\tilde{G})|/4$ leaves by applying the simplified case. However, since $|V(\tilde{G})|=|V(G)|$ and $T$ is also a spanning tree of $G$, so $T$ is also the required tree for the general claim. Hence in what follows we assume that $G$ has no edge with both endpoints of degree 2.

In order to build a spanning tree $T$ our algorithm begins with a tree consisting of an arbitrarily chosen vertex (called a {\em root}), and then the spanning tree is built by a sequence of {\em expansions}. To expand a leaf $v\in T$ means to add the vertices of $N_G(v)\setminus V(T)$ to $T$ and connect them to $v$ in $T$. Note that in a tree $T$ built from a root by a sequence of expansions, if a vertex in $V(G)-V(T)$ is adjacent with $v\in V(T)$, then $v$ is a leaf.

The order in which the leaves are expanded is important. To describe this order, we introduce three operations (operations O1 and O3 are the same as in~\cite{kw:many-leaves}, but O2 is modified):

\begin{enumerate}
 \item [(O1)] Applies when there is a leaf $v\in V(T)$ such that $|N_G(v)\setminus V(T)|\ge 2$. Then $v$ is expanded.
 \item [(O2)] Applies when there is a leaf $v\in V(T)$ such that $|N_G(v)\setminus V(T)|= 1$ (let $N_G(v)\setminus V(T) = \{x\}$), and moreover $|N_G(x)\setminus V(T)|=0$ or $|N_G(x)\cap V(T)|\ge 2$. Then $v$ is expanded.
 \item [(O3)] Applies when there is a leaf $v\in V(T)$ such that $|N_G(v)\setminus V(T)|= 1$ (let $N_G(v)\setminus V(T) = \{x\}$), and moreover $|N_G(x)\setminus V(T)|\ge 2$. Then $v$ is expanded and afterwards $x$ is expanded.
\end{enumerate}

Now we can describe the algorithm for Theorem~\ref{thm:n/4}, which we call GENERIC: 
\begin{enumerate}
\item choose an arbitrary vertex $r\in V$ and let $T=\{r\}$, 
\item apply O1-O3 as long as possible, giving precedence to O1.
\end{enumerate}
We claim that GENERIC returns a spanning tree of $G$.
Assume for a contradiction that at some point the algorithm is able to apply none of O1-O3 but $V(G)\ne V(T)$. Consider any leaf $v\in T$ such that $N_G(v) \not \subseteq V(T)$. Such a leaf exists because $V(G)\ne V(T)$ and $T$ is built by a sequence of expansions . Since O1 does not apply, $|N_G(v)\setminus V(T)|= 1$. Let  $N_G(v)\setminus V(T) = \{x\}$. Since O2 does not apply, $N_G(x)\cap V(T) = \{v\}$. Since neither O2 nor O3 apply, $|N_G(x)\setminus V(T)| = 1$. It follows that $\deg_G(x)=2$. Moreover, since there are no edges between 2-vertices, the neighbor of $x$ outside of $T$ is not a 2-vertex. It follows that $\bigcup_{v\in L(T)} N_G(v) \setminus V(T)$ is a separator consisting of 2-vertices, which is the desired contradiction.

It remains to show that if a spanning tree $T$ was constructed, then it has at least $n/4$ leaves. It can be done exactly as in the work of Kleitman and West~\cite{kw:many-leaves}. However, we do it in a different way, in order to introduce and get used to some notation that will be used in later sections, where we describe an improved kernel.

We say that a leaf $u$ of $T$ is {\em dead} if  $N_G(u) \setminus V(T) = \emptyset$.
Note that after performing O2 there is at least one new dead leaf: if $|N_G(x)\setminus V(T)|=0$ then $x$ is a dead leaf, and if $|N_G(x)\cap V(T)|\ge 2$ then all of $(N_G(x)\cap V(T))\setminus \{v\}$ are dead leaves, because of O1 precedence. For any tree $\hat{T}$, by $L(\hat{T})$ we denote the set of leaves of $\hat{T}$.

Let $X_i$ be the set of the inner vertices of $T$ that were expanded by an operation of type O$i$. Let $X$ be the set of the inner vertices of $T$; note that $X=X_1\cup X_2 \cup X_3$. Since $T$ is rooted, the standard notions of parent and children apply. For a positive integer $i$, let $P_i$ denote the set of vertices of $T$ with exactly $i$ children. 

Since every vertex besides $r$ is a child of some vertex, we have $ \sum_{d\ge 1} d|P_d| = n-1.$
\ignore{
\begin{equation}
\label{eq:children}
 \sum_{d\ge 1} d|P_d| = n-1.
\end{equation}
}
Since the set of vertices with one child is equal to $X_2\cup(X_3\cap P_1)$ and $|X_3\cap P_1|=|X_3\cap P_{\ge 2}|$ it follows that
\begin{equation}
\label{eq:children2}
 |X_2| + |X_3\cap P_{\ge 2}| + \sum_{d\ge 2} d|P_d| = n-1.
\end{equation}
Since during O2 at least one leaf dies, $|X_2| \le |L(T)|$.
Similarly, since after expanding a vertex from $X_3\cap P_{\ge 2}$ the cardinality of $L(T)$ increases, $|X_3\cap P_{\ge 2}| \le |L(T)|-1$.
Finally, $\sum_{d\ge 2} d|P_d| \le \sum_{d\ge 2} 2(d-1)|P_d| = 2(|L(T)|-1)$.
After plugging these three bounds to~\eqref{eq:children2} we get $|L(T)| > n/4$, as required. 
This finishes the proof of Theorem~\ref{thm:n/4}.

\subsection{Outerplanarity}
\label{sec:outer}

\begin{lemma}
\label{lem:outer}
If $G$ is a planar graph and $T$ is a spanning tree of $G$, then the graph $G[L(T)]$ is outerplanar.
\end{lemma}

\begin{proof}
Fix a plane embedding of $G$ and consider the induced plane subgraph $G' = G[L(T)]$. Since $T$ is connected, all vertices of $L(T)$ lie on the same face of $G'$. Therefore $G'$ is outerplanar.
\end{proof}

\begin{corollary}
If $G$ is a planar graph and $T$ is a spanning tree of $G$ then there is a subset of leaves of $T$ of size at least $|L(T)|/3$ that is independent in $G$ (and it can be found in linear time).
\end{corollary}

\begin{proof}
It is well-known that outerplanar graphs are 3-colorable and the 3-coloring can be found in linear time. So, by Lemma~\ref{lem:outer} we can 3-color $G[L(T)]$ and we choose the largest color class.  
\end{proof}

\section{A $\boldmath{9k}$ kernel}
\label{sec:9k}

In this section we present an improved kernel for the \indleafname\ problem. Although the analysis is considerably more involved than that of the $12k$ kernel, the algorithm is almost the same. We need only to force a certain order of the operations O1-O3 in step 2. As before, the algorithm always performs the O1 operation if possible (we will refer to this as {\em the O1 rule}). Moreover, if more than one O1 operation applies then we choose the one which maximizes the number of vertices added to $T$ (we will refer to this as {\em the largest branching rule}). If there is still more than one such operation applicable then among them we choose the one which expands a vertex that was added to $T$ later than the vertices which would be expanded by other operation (we will refer to this as {\em the DFS rule}). Similarly, if there are no O1 operations applicable but more than one O2/O3 operations apply, we also use the DFS rule. 
The algorithm GENERIC with the order of operations described above will be called BRANCHING. 
%Since it is just a special case of GENERIC, so all the claims we proved in Section~\ref{sec:12k} apply. 

Note that the algorithm BRANCHING is just a special case of GENERIC, so all the claims we proved in Section~\ref{sec:12k} apply. 
Let us think where the bottleneck in this analysis is. There are two sources of trouble. First, if there are many O2/O3 operations we get a spanning tree with few leaves: in particular there might be only O2 operations and O3 operations that add just two leaves (consider a cubic graph obtained from the cycle $v_0v_1\ldots v_{3k-1}v_0$ by adding $k$ vertices $w_0,\ldots,w_{k-1}$ and for each $i=0,\ldots,k-1$ add the following edges: $w_iv_{3i}, w_iv_{3i+1}, w_iv_{3i+2}$) and we get roughly $n/4$ leaves. Second, if the outerplanar graph $G[L(T)]$ is far from being bipartite (i.e. has many short odd cycles) then we get a small independent set: in particular, when $G[L(T)]$ is a collection of disjoint triangles, the maximum independent set in $G[L(T)]$ is of size exactly $|L(T)|/3$. However, we will show that these two extremes cannot happen simultaneously. More precisely, we prove the following two theorems.

\begin{theorem}
\label{thm:many-leaves}
Let $G$ be a connected $n$-vertex graph that does not contain a separator consisting of only 2-vertices. Then $G$ has a  spanning tree $T$ such that if $\C$ is a collection of vertex-disjoint cycles in $G[L(T)]$, then
\[|L(T)| \ge \frac{n+3|\C|}{4}.\]
Moreover, $T$ can be found in linear time.
\end{theorem}

\begin{theorem}
\label{thm:big-is}
Every $n$-vertex outerplanar graph contains
\begin{itemize}
 \item an independent set $I$, and
 \item a collection of vertex-disjoint cycles $\C$
\end{itemize}
such that $9|I|\ge 4n-3|\C|$.
\end{theorem}

Note that the bound in Theorem~\ref{thm:big-is} is tight, which is easy to see by considering an outerplanar graph consisting of disjoint triangles. From the above two theorems and Lemma~\ref{lem:outer} we easily get the following corollary.
\begin{corollary}
Let $G$ be a connected $n$-vertex graph that does not contain a separator consisting of only 2-vertices. Then $G$ has a  spanning tree $T$ such that $L(T)$ has a subset of size at least $n/9$ which is independent in $G$. Equivalently, $G$ has a nonseparating independent set of size at least $n/9$ and a connected vertex cover of size at most $\frac{8}{9}n$.
\end{corollary}

By a similar reasoning as in the beginning of Section~\ref{sec:12k} we get a $9k$-kernel for the \indleafname\ problem. In what follows, we prove Theorem~\ref{thm:many-leaves} and Theorem~\ref{thm:big-is}.
%in Appendix~\ref{app:thm:big-is}.

\subsection{Proof of Theorem~\ref{thm:many-leaves}}

Note that similarly as in Theorem~\ref{thm:n/4} it suffices to show a simplified case when $G$ has no edge $uv$ such that both $u$ and $v$ are 2-vertices. Indeed, if the theorem holds for the simplified case, as before we create a new graph $\tilde{G}$  by removing from $G$ all the edges with both endpoints of degree 2 and as before $\tilde{G}$ does not contain a separator consisting of only 2-vertices. Then we apply the simplified case and we get a spanning tree $T$ such that for any collection  $\C$ of vertex-disjoint cycles in $\tilde{G}[L(T)]$, we have $|L(T)| \ge (|V(\tilde{G})|+3|\C|)/4$ and $T$ is a spanning tree of $G$ as well. Moreover, no edge of $E(\tilde{G})\setminus E(G)$ belongs to a cycle in $G[L(T)]$ for otherwise both of its endpoints have degree at least 3 in $G$. Hence if $\C$ is a collection of vertex-disjoint cycles in $G[L(T)]$ then it is also a collection of vertex-disjoint cycles in $\tilde{G}[L(T)]$, so the desired inequality holds.
Hence in what follows we assume that $G$ has no edge with both endpoints of degree 2. Since we proved that in this case the algorithm GENERIC returns a spanning tree, and each execution of BRANCHING is just a special case of an execution of GENERIC we infer that BRANCHING returns a spanning tree of $G$, which will be denoted by $T$.

Let $\C$ be an arbitrary collection of vertex-disjoint cycles in $G[L(T)]$. Our general plan for proving the claim of Theorem~\ref{thm:many-leaves} is to show that if $|\C|$ is large then we have few O2/O3 operations --- by~\eqref{eq:children2} this will improve our bound on $|L(T)|$. To be more precise, let us introduce several definitions.

Recall the O2 operation: it adds a single vertex $x$ to $T$ and at least one leaf of $T$ dies. We choose exactly one of these dead leaves and we {\em assign} it to $x$. However, if the vertex $x$ dies we always assign $x$ to itself (so if some other leaves die during this operation, they are unassigned). Let $L_u$ be the set of unassigned leaves of $T$. Clearly, $|X_2|=|L(T)|-|L_u|$. In order to show that there are few O2 operations, we will show that $|L_u|$ is big.

Let $x_1,x_2,\ldots x_{|X|}$ be the inner vertices of $T$ in the order of expanding them (in particular $x_1=r$). 
A {\em run} is a maximal subsequence $x_b,x_{b+1},\ldots,x_e$ of vertices from $P_{\ge 2}$, i.e.\ the nodes in $T$ that have at least two children.

\begin{lemma}
\label{lem:subtree}
 Vertices of any run $R=x_b,\ldots,x_e$ form a subtree of $T$ rooted at $x_b$.
\end{lemma}

\begin{proof}
Assume that a vertex $x\in \{x_{b+1},\ldots,x_e\}$ has the parent $x_p$ outside the run. Then $p<b$ and $x$ was a leaf in $T$ while $x_b$ was being expanded. Hence, by the definition of a run, $x_{b-1}\in P_1$, and in particular $x_{b-1}$ was expanded by O2 or O3. However, when this operation was performed, it was possible to expand $x$ by O1, a contradiction with the O1 rule. Hence every vertex $x\in \{x_{b+1},\ldots,x_e\}$ has the parent in the run, which is equivalent to the claim of the lemma. 
\end{proof}

In what follows, the subtree from Lemma~\ref{lem:subtree} is denoted by $T_R$. 
Moreover, let $\ch(T_R)$ denote the set of children of the leaves of $T_R$, i.e.\ \[\ch(T_R)=\{v \in V(T)\setminus V(T_R)\ :\ \text{$T_R$ contains the parent of $v$}\}.\]
We say that a run $R$ {\em opens} a cycle $C$ in $\C$ if the first vertex of $C$ that was added to $T$ belongs to $\ch(T_R)$. The following lemma shows a relation between cycles in $\C$ and runs.

\begin{lemma}
\label{lem:cycle-run}
Every cycle in $\C$ is opened by some run.
\end{lemma}

\begin{proof}
Consider any cycle $C\in\C$ and let $v$ be the first vertex of $C$ that is added to $T$. Note that $v$ is not added by O2, for otherwise just after adding $v$ to $T$, $v$ has at least two neighbors and by the O1 rule $v$ would be the next vertex expanded and hence not a leaf of $T$, a contradiction.
It follows that $v$ is added by O1 or O3 and consequently $v\in \ch(T_R)$ for some run $R$. 
\end{proof}

Now we can sketch our idea for bounding the number of O3 operations (\#O3). Both after O1 and O3 the cardinality of $L(T)$ increases. Hence, if we fix the number of leaves in the final tree, then if $|X_1|$ is large then \#O3 should be small. Since a run contains at most one vertex of $|X_3|$ (e.g.\ by Lemma~\ref{lem:subtree}), it means that a tree $T_R$ with a large number of children contains plenty of vertices from $|X_1|$. We will show that if a run opens many cycles, then indeed $|\ch(T_R)|$ is large. Let $\C_R$ denote the set of cycles in $\C$ opened by $R$. 

\begin{lemma}
 \label{lem:one-cycle}
 Let $R$ be a run. For any cycle $C\in\C_R$ one of the following conditions holds:
 \begin{enumerate}[$(i)$]
  \item $|\ch(T_R)\cap V(C)| + |L_u\cap V(C)| \ge 4$, or
  \item $|\ch(T_R)\cap V(C)| + |L_u\cap V(C)| = 3$ and $|R\cap P_{\ge 3}|\ge 1$.
 \end{enumerate}
\end{lemma}

\begin{proof}
Let $v_1$ be the vertex of $C$ that is added first to the tree $T$.
By the definition of $\C_R$, $v_1\in \ch(T_R)$.
We see that at least one neighbor of $v_1$, call it $v_2$, is in $\ch(T_R)$, for otherwise just after expanding the last vertex of $R$ the vertex $v_1$ can be expanded by O1, while the algorithm chooses O2/O3, a contradiction with the O1 rule.

Let $w$ be the neighbor of $v_2$ on $C$ that is distinct from $v_1$. 
Assume $w\not\in\ch(T_R)$. Then just after expanding the last vertex of $R$ we have $N(v_2)\setminus T = \{w\}$, since if $|N(v_2)\setminus T| \ge 2$ then it is possible to expand $v_2$ by O1. Hence if $v_2$ is assigned then it is assigned to $w$. However, then $w$ is added to $T$ by O2 so 
$w$ dies during this operation (otherwise $w$ is expanded because of the DFS rule so $w\not\in L(T)$), and hence $w$ is assigned to $w$ and $v_2\in L_u$. To conclude, $w\in\ch(T_R)$ or $v_2\in L_u$.

If we denote by $u$ the neighbor of $v_1$ on $C$ that is distinct from $v_2$, by the same argument we get $u\in\ch(T_R)$ or $v_1\in L_u$.

It follows that $(i)$ holds, unless $u=w$ (i.e.\ $C$ is a triangle), $v_1,v_2 \not\in L_u$ and $w\in\ch(T_R)$. Let us investigate this last case.
We see that $|\ch(T_R)\cap V(C)|=3$. We will show that $|R\cap P_{\ge 3}|\ge 1$.
Since $v_2,w\in\ch(T_R)$, they could not be added by O2 and hence $v_1$ is assigned to a vertex $x\not\in V(C)$.
Note that $x$ is added to $T$ after $v_2$ and $w$.
Assume w.l.o.g. that $v_2$ was added to $T$ before $w$. 
We consider two cases.
If $v_2$ was not added to $T$ by expanding the parent of $v_1$, then the parent $p$ of $v_2$ has at least three children (otherwise instead of expanding $p$ the algorithm can expands $v_1$ and add at least three children, a contradiction with the largest branching rule), so $|R\cap P_{\ge 3}|\ge 1$ as required.
Finally, if $v_2$ was added to $T$ by expanding the parent $p$ of $v_1$, then $p\in P_{\ge 3}$ for otherwise just after expanding $p$ O1 is applicable to $v_1$ so either $v_1$ or $v_2$ is expanded by the DFS rule. This concludes the proof.
\end{proof}

By applying Lemma~\ref{lem:one-cycle} to all cycles of a single run $R$ we get the following corollary.

\begin{corollary}
\label{cor:run}
For any run $R$ that opens at least one cycle, 
\[|\ch(T_R)\cap V(\C_R)| + |L_u\cap V(\C_R)| + |R\cap P_{\ge 3}| \ge 3|\C_R|+1.\]

\end{corollary}

%\begin{proof}
% Apply Lemma~\ref{lem:one-cycle} to all cycles of $R$. 
%\end{proof}

\begin{lemma}
\label{lem:run}
 For any run $R$,
 \begin{equation}
 \label{eq:run} 
 |L_u \cap V(\C_R)| + \sum_{d \ge 2}(2d-3)|R\cap P_d| -|R\cap X_3| \ge 3|\C_R|.
 \end{equation}
\end{lemma}

\begin{proof}
\case{1}: $\ch(T_R) \subseteq V(\C_R)$.
We claim that then $V(\C_R) \subseteq \ch(T_R)$ and so $V(\C_R) =\ch(T_R)$.
If there is a cycle $C\in\C_R$ with a vertex outside $\ch(T_R)$ then $C$ has a vertex in $\ch(T_R)$ with a neighbor in $V(C)\setminus\ch(T_R)$.
Among all such vertices for all cycles of $\C_R$ choose the one that was added to $T$ last, call it $v$. We see that after expanding all vertices of $R$ the algorithm expands $v$ (by the DFS rule and because $\ch(T_R) \subseteq V(\C_R)$), a contradiction because $v\in L(T)$. Hence indeed $V(\C_R) \subseteq \ch(T_R)$.

Now assume that there is $v\in V(\C_R)$ such that $v\not\in L_u$. Among all such vertices choose the one that was added last to $T$, call it $u$.
Note that $u$ is assigned to a vertex that was added to $T$ after the run $R$. We get a contradiction again, since after expanding all vertices of $R$ the algorithm expands $u$. To conclude, $V(\C_R)\subseteq L_u$, hence $|L_u\cap V(\C_R)| = \sum_{C\in\C_R}|V(C)| \ge 3|\C_R|.$ Since $|R\cap P_{\ge 2}|\ge |R\cap X_3|$, the claim follows.

\case{2}: $\ch(T_R) \not\subseteq V(\C_R)$.
Then, $|\ch(T_R)\cap V(\C_R)| \le |\ch(T_R)|-1.$
Since $|\ch(T_R)| = \sum_{d\ge 2}d|R\cap P_d| - (|R|-1)$, we get 
\begin{equation}
 \label{eq:T_d}
 \sum_{d\ge 2}|R\cap P_d|(d-1) \ge |\ch(T_R)\cap V(\C_R)|.
\end{equation}
Now observe that the LHS of~\eqref{eq:run} is always nonnegative, so the claim holds when $\C_R=\emptyset$.
Hence we can assume that $\C_R\ne\emptyset$ and the bound of Corollary~\ref{cor:run} applies.
Note also that $|R\cap X_3|\le 1$, which follows e.g.\ from Lemma~\ref{lem:subtree}.
Then,
\begin{equation*}
\begin{split}
|L_u \cap V(\C_R)| + \sum_{d\ge 2}|R\cap P_d|(2d-3) -|R\cap X_3| & \ge \\
|L_u \cap V(\C_R)| + \sum_{d\ge 2}|R\cap P_d|(d-1) + |R\cap P_{\ge 3}|- 1 & \ge^{\text{\eqref{eq:T_d}}} \\
|L_u \cap V(\C_R)| + |\ch(T_R)\cap V(\C_R)| + |R\cap P_{\ge 3}|- 1  & \ge^{\text{(Corollary~\ref{cor:run}})} \\
3|\C_R|.
\end{split}
\end{equation*}
\end{proof}

Now we are ready to prove the claim of Theorem~\ref{thm:many-leaves}, i.e.\ that $|L(T)|\ge (n+3|\C|)/4$.

Let  us add $\sum_{d\ge 2}(2d-3)|P_d|$ to both sides of~\eqref{eq:children2}:
\begin{equation}
|X_2| + |X_3\cap P_{\ge 2}| + 3\sum_{d\ge 2} (d-1)|P_d| = n-1+\sum_{d\ge 2}(2d-3)|P_d|.
\end{equation}
Since $|X_2| = |L(T)|-|L_u|$ and $\sum_{d\ge 2} (d-1)|P_d|=|L(T)|-1$ we get
\begin{equation}
\label{eq:dupa1}
4|L(T)| = n+2+|L_u|+\sum_{d\ge 2}(2d-3)|P_d|-|X_3\cap P_{\ge 2}|.
\end{equation}
By Lemma~\ref{lem:cycle-run} after summing~\eqref{eq:run} over all runs we get
\begin{equation}
\label{eq:dupa2}
|L_u| + \sum_{d\ge 2}(2d-3)|P_d| -|X_3\cap P_{\ge 2}| \ge 3|\C|.
\end{equation}
The claim follows immediately after plugging~\eqref{eq:dupa2} to~\eqref{eq:dupa1}.

\subsection{Proof of Theorem~\ref{thm:big-is}}
\label{app:thm:big-is}

\newtheorem*{thm:big-is:restated}{Theorem \ref{thm:big-is} (restated)}

\begin{thm:big-is:restated}
Every $n$-vertex outerplanar graph contains
\begin{itemize}
 \item an independent set $I$, and
 \item a collection of vertex-disjoint cycles $\C$
\end{itemize}
such that $9|I|\ge 4n-3|\C|$.
\end{thm:big-is:restated}

\begin{proof}
Let $H=(V,E)$ be the outerplanar graph under consideration.
We use the induction on $|V|$.
If $|V|=0$ we put $I=\C=\emptyset$ and the claim follows.
Hence we can assume that $H$ has at least one vertex.

\case{1}: There is a vertex $v$ such that $\deg_{H}v \le 1$.

Let $H'$ be the graph obtained from $H$ by removing $v$ and its neighbor, if any.
We apply induction to $H'$ and we get relevant $I_0$, $\C_0$ such that
\begin{equation}
 9|I_0|\ge 4n - 8 - 3|\C_0|.
\end{equation}
We set $I=I_0\cup \{v\}$ and $\C=\C_0$. Clearly, $I$ is an independent set and we have
\[ 9|I| = 9(|I_0|+1) \ge 4n+1-3|\C|.\]
This finishes Case 1.

Hence, in what follows we assume that every vertex of $H$ has degree at least 2. 
Now, consider a leaf block $Q$ of $H$, i.e.\ a block that contains at most one cutvertex of $H$.
Note that $Q$ is not a single edge, since then Case 1 applies.
Fix any outerplanar embedding of $Q$ and let $\bar{f}$ be the outer face of $Q$ in this embedding.
Let $D(Q)$ be the graph dual to $Q$ and let $T_Q$ be the graph obtained from $D(Q)$ by removing vertex $\bar{f}$.
Since $Q$ is outerplanar, $T_Q$ is a tree (in particular $T_Q$ is simple).
A face of $Q$ that is a leaf of $T_Q$ will be called a leaf-face.

\case{2}: Block $Q$ contains a leaf-face of odd length $f=v_1v_2\ldots v_{\ell}$ such that for $i=2,\ldots,\ell-1$ we have $\deg_H(v_i)=2$.

We apply induction to $H-V(f)$ and we get relevant $I_0$, $\C_0$ such that
\begin{equation}
 9|I_0|\ge 4n - 4\ell - 3|\C_0|.
\end{equation}
We set $I=I_0\cup \{v_2,v_4,\ldots,v_{\ell-1}\}$ and $\C=\C_0\cup\{f\}$. 
Note that since $Q$ is not a single edge, $\deg_Q(v)\ge 2$ for every vertex $v\in Q$, so if $\deg_H(v)=2$ then $v$ is not a cut vertex in $H$.
In particular none of the vertices $v_2,v_4,\ldots,v_{\ell-1}$ is a cutvertex.
It follows that $I$ is an independent set and we have
\[ 9|I| = 9\left(|I_0|+\frac{\ell-1}{2}\right) \ge 4n - 3|\C_0| + \frac{9}{2}(\ell-1)-4\ell=4n-3|\C|+\frac{\ell-3}{2}\ge 4n-3|\C|.\]

\case{3}: Block $Q$ contains a leaf-face of even length $f=v_1v_2\ldots v_{\ell}$ such that for $i=2,\ldots,\ell-1$ we have $\deg_H(v_i)=2$ and $\deg_H(v_{\ell})=\deg_Q(v_{\ell})=3$.

Let $w$ be the neighbor of $v_{\ell}$ distinct from $v_1,\ldots,v_{\ell-1}$.
We apply induction to $H-(V(f)\cup\{w\})$ and we get relevant $I_0$, $\C_0$ such that
\begin{equation}
 9|I_0|\ge 4n - 4\ell -4- 3|\C_0|.
\end{equation}
We set $I=I_0\cup \{v_2,v_4,\ldots,v_{\ell}\}$ and $\C=\C_0\cup\{f\}$. 
Clearly, $I$ is an independent set and we have
\[ 9|I| = 9\left(|I_0|+\frac{\ell}{2}\right) \ge 4n - 3|\C_0| + \frac{9}{2}\ell-4\ell-4=4n-3|\C|+\frac{\ell-2}{2}> 4n-3|\C|.\]

\case{4}: Block $Q$ contains two leaf-faces of even length with a common vertex, $f_1=v_1v_2\ldots v_{\ell}$ and $f_2=w_1w_2\ldots w_{k}$ with $v_{\ell}=w_k$, such that for $i=2,\ldots,\ell-1$ we have $\deg_H(v_i)=2$, for $i=2,\ldots,k-1$ we have $\deg_H(v_i)=2$ and $\deg_H(v_{\ell})=\deg_Q(v_{\ell})=4$.

We apply induction to $H-(V(f_1)\cup V(f_2))$ and we get relevant $I_0$, $\C_0$ such that
\begin{equation}
 9|I_0|\ge 4n - 4(k+\ell-1)- 3|\C_0|.
\end{equation}
We set $I=I_0\cup \{v_2,v_4,\ldots,v_{\ell}\}\cup \{w_2,w_4,\ldots,w_{k}\}$ and $\C=\C_0\cup\{f_1\}$. 
Clearly, $I$ is an independent set and we have
\[ 9|I| = 9\left(|I_0|+\frac{\ell}{2}+\frac{k}{2}-1\right) \ge 4n - 3|\C_0| + \frac{\ell+k}{2}-5=4n-3|\C|+\frac{\ell+k-4}{2}> 4n-3|\C|.\]

\case{5}: none of Cases 1-4 applies.

We show that $Q$ has at most two bounded faces.

First, assume that $T_Q$ has at least two non-leaves. Then, there is a non-leaf $f^*$ such that
\begin{itemize}
 \item $f^*$ has exactly 1 non-leaf neighbor $f^{**}$ in $T_Q$, and
 \item for every leaf neighbor $f$ of $f^*$, if $f$ is incident with a cutvertex $x$ in $H$, then $x$ is incident also with both $f^*$ and $f^{**}$.
\end{itemize}
Note that by Case 2 all leaf neighbors of $f^*$ are of even length.
Moreover, by Case 3, $f^*$ has $|f^*|-1\ge 2$ leaf neighbors.
Hence, Case 4 applies, a contradiction.

Now, assume $T_Q$ has exactly one non-leaf $f^*$.
Then, $f^*$ has at least 2 leaf neighbors in $T_Q$.
By Case 2, at least one of them, call it $f$, is of even length (note that $f^*$ may have a leaf neighbor $f'$ of odd length which contains a vertex $v$ such that $v$ is not incident with $f^*$ and $v$ is a cutvertex of $H$: then $\deg_H(v)>2$ and Case 2 would not apply to $f'$, but if there is another odd leaf neighbor of $f^*$ then Case 2 applies to it).
Since Case 3 does not apply to $f$, so (a) $f$ is incident with another two leaf neighbors $f'$ and $f''$ of $f^*$ or (b) $f$ is incident with exactly one leaf neighbor $f'$ and $f$ has a cutvertex $x$ of $H$ such that $x$ is not incident with $f'$.
In situation (a), by Case 2 and since $Q$ has at most one cutvertex of $H$, at least one of these leaf neighbors is of even length. But then Case 4 applies, a contradiction.
In situation (b), $f'$ is of even length by Case 2 and hence Case 4 applies to $f$ and $f'$, a contradiction.

It follows that $T_Q$ has only leaves, i.e.\ $T_Q$ is a single leaf or two adjacent leaves (in other words $Q$ has at most two bounded faces).

\case{5a}: $|V(T_Q)|=1$, i.e.\ $Q$ is a single facial cycle $C$.

Note that $|C|$ is even for otherwise Case 2 applies.
Then we apply induction to $H-V(C)$ and we get relevant $I_0$, $\C_0$ such that
\begin{equation}
 9|I_0|\ge 4n - 4|C|- 3|\C_0|.
\end{equation}
We set $\C=\C_0$ and we set $I$ to be the independent set obtained from $I_0$ by adding every second vertex of $C$: this can be done in two ways and we choose a way that does not add a cutvertex of $H$. 
\[ 9|I| = 9\left(|I_0|+\frac{|C|}{2}\right) \ge 4n -4|C|- 3|\C| + \frac{9}{2}|C|=4n-3|\C|+\frac{|C|}{2}> 4n-3|\C|.\]

\case{5b}: $|V(T_Q)|=2$, i.e.\ $Q$ has two bounded faces $f_1$, $f_2$ and they share an edge.

Since Case 2 does not apply, at least one face, say $f_1$ is of even length.
If both $f_1$ and $f_2$ are even then we apply induction to $H-V(Q)$ getting relevant $I_0$ and $\C_0$ and we create $I$ by extending $I_0$ by every second vertex of the outer face of $Q$ so that a cutvertex of $H$, if any, is not added. By the same calculation as in Case 5a we get the desired inequality.
Hence we can assume that $f_1$ is even and $f_2$ is odd, but we will show that this last case cannot happen. Then a cutvertex of $H$ belongs to $V(f_1)$, for otherwise Case 3 applies. But then Case 2 applies, a contradiction. This concludes the proof of Theorem~\ref{thm:big-is}.
\end{proof}

\section{A simple $5k$ kernel for {\sc (Planar) Max Leaf}}

In this section we show a simple kernelization algorithm for \maxleafname\, and \pmaxleafname. Below we describe three simple rules, which preserve planarity.

\begin{itemize}
 \item {\bf $(1,2)$-rule} If there is a 1-vertex $u$ adjacent with a 2-vertex $v$ then remove $v$.
 \item 
 {\bf Adjacent 2-vertices Rule} 
 Assume that there are two adjacent 2-vertices $u$ and $v$.
 If $uv$ is a bridge, contract $uv$, otherwise remove $uv$.
 \item {\bf Trivial Rule} If $G$ consists of a single edge, return YES if $k\le 2$.
\end{itemize}

It is quite clear that the above rules are correct for ({\sc Planar}) \maxleafname\; (see e.g.~\cite{prieto-phd}, Rules 1-3 for a proof).
Note that if none of our rules applies to a connected graph $G$, then every edge of $G$ has an endpoint of degree at least 3.

\begin{theorem}
\label{thm:n/5}
Let $G$ be a connected graph in which every edge has an endpoint of degree at least 3. Then $G$ has a spanning tree with at least $n/5$ leaves.  
\end{theorem}

\begin{proof}
 We proceed similarly as in the proof of Theorem~\ref{thm:n/4}. As before, we start from a single vertex $r$ and we build a tree $T$ by applying operations O1-O3 as long as possible. Then we perform a new operation O4, which expands any leaf $v$ of $T$ which is not dead. Hence if $T$ is not spanning, one of operations O1-O4 is applicable. The algorithm builds the tree by applying O1-O4, but O4 is performed only if none of O1-O3 applies. Let $T$ be the final spanning tree of $G$.
 
 Now we show a lower bound on $|L(T)|$. Using the notation from Section~\ref{sec:many-leaves} we can state an analog of~\eqref{eq:children2}:
\begin{equation}
\label{eq:chyba-ostatnie}
|X_2| + |X_3\cap P_{1}| + |X_4| + \sum_{d\ge 2} d|P_d| = n-1.
\end{equation}
 Consider an operation O4, which expands a leaf $v$ of $T$.
 As we have shown in Section~\ref{sec:many-leaves}, $|N_G(v)\setminus V(T)|= \{x\}$ for some $x$ of degree 2 in $G$ and moreover the neighbor $z$ of $x$ distinct from $v$ is outside $T$. Since every edge has an endpoint of degree at least 3, $\deg(z)\ge 3$. It follows that just after expanding $v$ by O4, $x$ is expanded by operation O3. Hence $|X_4|\le |X_3 \cap P_1|$.
 Moreover, in Section~\ref{sec:many-leaves} we have shown that $|X_3 \cap P_1| = |X_3 \cap P_{\ge 2}| \le |L(T)|-1$ and $|X_2|\le |L(T)|$. 
After plugging these three bounds to~\eqref{eq:chyba-ostatnie} we get $|L(T)|> n/5$, as required.
\end{proof}

Theorem~\ref{thm:n/5} is another variation of the Kleitman-West result which might be of independent interest.
Note that the bound of Theorem~\ref{thm:n/5} is tight up to an additive constant, which is shown by the following family of graphs $G_k$. Begin with a cycle $v_0v_1\ldots v_{4k-1}v_0$. Then add $k$ vertices $w_0,\ldots,w_{k-1}$. Finally, for each $i=0,\ldots,k-1$ add the following edges: $w_iv_{4i+1}, w_iv_{4i+2}, w_iv_{4i+3}$. It is easy to see that every edge of $G_k$ has an endpoint of degree at least 3. It is also easy to see that if $T_k$ is a spanning tree of $G_k$ with maximum possible number of leaves, then $\lim |L(T_k)| / |V(G_k)| = 5$.

Now we can describe our kernelization algorithm.
Let $G'$ be the graph obtained from $G$ by applying our three rules as long as one of them applies.
By Theorem~\ref{thm:n/5}, if $k\le n/5$ we can return the answer YES. Hence $n < 5k$ and $G'$ is a $5k$-kernel for \pmaxleafname\ and \maxleafname.

\myparagraph{Acknowledgments} We are very grateful for the reviewers for numerous comments.
We also thank Michal Debski for helpful discussions.

\bibliographystyle{abbrv}
%\bibliography{dual-planar-cvc}

\begin{thebibliography}{10}

\bibitem{afn:planar-domset}
J.~Alber, M.~R. Fellows, and R.~Niedermeier.
\newblock Polynomial-time data reduction for dominating set.
\newblock {\em J. ACM}, 51(3):363--384, 2004.

\bibitem{cfkx:duality-and-vertex}
J.~Chen, H.~Fernau, I.~A. Kanj, and G.~Xia.
\newblock Parametric duality and kernelization: Lower bounds and upper bounds
  on kernel size.
\newblock {\em SIAM J. Comput.}, 37(4):1077--1106, 2007.

\bibitem{fellows:maxleaf}
V.~Estivill-Castro, M.~R. Fellows, M.~A. Langston, and F.~A. Rosamond.
\newblock {FPT is P-Time Extremal Structure I}.
\newblock In {\em ACiD 2005}, pages 1--41, 2005.

\bibitem{fomin:bidim-kernels}
F.~V. Fomin, D.~Lokshtanov, S.~Saurabh, and D.~M. Thilikos.
\newblock Bidimensionality and kernels.
\newblock In M.~Charikar, editor, {\em SODA}, pages 503--510. SIAM, 2010.

\bibitem{garey:planarcvc}
M.~R. Garey and D.~S. Johnson.
\newblock The rectilinear steiner tree problem in {NP} complete.
\newblock {\em SIAM Journal of Applied Mathematics}, 32:826--834, 1977.

\bibitem{guon:planarkernels}
J.~Guo and R.~Niedermeier.
\newblock Linear problem kernels for {NP}-hard problems on planar graphs.
\newblock In {\em Proc. ICALP'07}, volume 4596 of {\em Lecture Notes in
  Computer Science}, pages 375--386, 2007.

\bibitem{huang}
Y.~Huang and Y.~Liu.
\newblock Maximum genus and maximum nonseparating independent set of a
  3-regular graph.
\newblock {\em Discrete Mathematics}, 176(1-3):149 -- 158, 1997.

\bibitem{kw:many-leaves}
D.~J. Kleitman and D.~B. West.
\newblock Spanning trees with many leaves.
\newblock {\em SIAM J. Discrete Math.}, 4(1):99--106, 1991.

\bibitem{moja-cvc-11/3}
L.~Kowalik, M.~Pilipczuk, and K.~Suchan.
\newblock Towards optimal kernel for connected vertex cover in planar graphs.
\newblock {\em CoRR}, abs/1110.1964, 2011.
\newblock To appear in Discr. Appl. Math.

\bibitem{prieto-phd}
E.~Prieto-Rodriguez.
\newblock {\em {Systematic kernelization in FPT algorithm design}}.
\newblock PhD thesis, University of Newcastle, 2005.

\bibitem{Speckenmeyer}
E.~Speckenmeyer.
\newblock On feedback vertex sets and nonseparating independent sets in cubic
  graphs.
\newblock {\em Journal of Graph Theory}, 12(3):405--412, 1988.

\bibitem{ueno}
S.~Ueno, Y.~Kajitani, and S.~Gotoh.
\newblock On the nonseparating independent set problem and feedback set problem
  for graphs with no vertex degree exceeding three.
\newblock {\em Discrete Mathematics}, 72(1–3):355 -- 360, 1988.

\bibitem{mfcs}
J.~Wang, Y.~Yang, J.~Guo, and J.~Chen.
\newblock Linear problem kernels for planar graph problems with small distance
  property.
\newblock In {\em Proc. MFCS'11}, volume 6907 of {\em LNCS}, pages 592--603,
  2011.

\end{thebibliography}

\end{document}